\documentclass[preprint]{elsarticle}

\usepackage{amsfonts}
\usepackage{amsbsy}
\usepackage{amssymb}
\usepackage{amscd}
\usepackage[cmex10]{amsmath}
\usepackage{graphicx}
\usepackage[ruled,vlined,linesnumbered]{algorithm2e}

\newtheorem{definition}{Definition}
\newtheorem{theorem}{Theorem}
\newproof{proof}{Proof}

\newcommand{\mat}[1]{\boldsymbol{#1}}

\newcommand{\pp}[1]{{\left( #1 \right)}}

\newcommand{\br}[1]{{\{ #1 \}}}

\newcommand{\norm}[1]{{ \Vert #1 \Vert }}
\newcommand{\abs}[1]{{ | #1 | }}
\newcommand{\sabs}[1]{{ | #1|^2 }}
\newcommand{\snorm}[1]{{ \Vert #1 \Vert^2 }}

\DeclareMathOperator*{\argmax}{arg\,max}

\def\mrt{{\text{\tiny{MRT}}}}
\def\tsum{{\text{\tiny{sum}}}}
\def\zf{{\text{\tiny{ZF}}}}  % ZF
\def\rmrt{{\text{\tiny{R-MRT}}}}% robust MRT
\def\H{{H}} % Hermitian
\def\ld{{\log_2}} % Logarithm dualis
\def\bI{{\mat{I}}} % Identity
\def\bA{{\mat{A}}} % unitary matrix

\def\bh{{\mat{h}}} % channel vector h
\def\ebh{{\tilde{\mat{h}}}} % estimate channel vector h
\def\bw{{\mat{w}}} % beamforming vector w
\def\ebno{{\frac{E_b}{N_0}}}
\def\ebnof{{{E_b}/{N_0}}}
\def\cC{{C}}
\def\sC{{\mathbf{C}}}
\def\pCSI{{\text{pCSI}}}
\def\iCSI{{\text{iCSI}}}
\def\SNR{{\rho}}

\journal{Signal Processing}

\begin{document}

\begin{frontmatter}

% ################################### --------- Title

\title{Robust Beamforming in Interference Channels with Imperfect Transmitter Channel Information\tnoteref{label1}}%

\tnotetext[label1]{Part of this work has been performed in the framework of the European research project SAPHYRE, which is partly funded by the European Union under its FP7 ICT Objective 1.1 - The Network of the Future. This work is also supported in part by the Deutsche Forschungsgemeinschaft (DFG) under grant Jo 801/4-1.}

\author{Rami Mochaourab\fnref{corAuth}}\ead{Rami.Mochaourab@tu-dresden.de}
\author{Eduard A. Jorswieck}\ead{Eduard.Jorswieck@tu-dresden.de}
\address{Department of Electrical Engineering and Information Technology \\Dresden University of Technology, 01062 Dresden, Germany. \\ Phone: +49-351-46332239. Fax: +49-351-46337236.}

\fntext[corAuth]{Corresponding author}

\begin{abstract}
We consider $K$ links operating concurrently in the same spectral band. Each transmitter has multiple antennas, while each receiver uses a single antenna. This setting corresponds to the multiple-input single-output interference channel. We assume perfect channel state information at the single-user decoding receivers whereas the transmitters only have estimates of the true channels. The channel estimation errors are assumed to be bounded in elliptical regions whose geometry is known at the transmitters. Robust beamforming optimizes worst-case received power gains, and a Pareto optimal point is a worst-case achievable rate tuple from which it is impossible to increase a link's performance without degrading the performance of another. We characterize the robust beamforming vectors necessary to operate at any Pareto optimal point. Moreover, these beamforming vectors are parameterized by $K(K-1)$ real-valued parameters. We analyze the system's spectral efficiency at high and low signal-to-noise ratio (SNR). Zero forcing transmission achieves full multiplexing gain at high SNR only if the estimation errors scale linearly with inverse SNR. If the errors are SNR independent, then single-user transmission is optimal at high SNR. At low SNR, robust maximum ratio transmission optimizes the minimum energy per bit for reliable communication. Numerical simulations illustrate the gained theoretical results.
\end{abstract}


\end{frontmatter}

\section{Introduction}
We consider multiple transmitter-receiver pairs simultaneously operating on the same frequency band. The signal from a transmitter is useful information at the intended receiver while it is  regarded as interference at unintended receivers. All transmitters are equipped with multiple antennas while the receivers use only a single antenna. This setting corresponds to the multiple-input single-output (MISO) interference channel (IFC) \cite{Vishwanath2004}.

In the MISO IFC, the beamforming vectors used at the transmitters impact the performance of the systems. A jointly efficient operating point corresponds to a Pareto optimal point in which it is not possible to improve the performance of one link without degrading the performance of at least another link. Designing a Pareto optimal mechanism requires finding the joint beamforming vectors used at the transmitters that lead to the Pareto optimal point. In the MISO IFC, finding specific Pareto optimal points such as the maximum sum-rate or proportional-fair is proven to be strongly NP-hard \cite{Luo2008,Liu2011}. The importance of characterizing the set of beamforming vectors necessary for the links' Pareto optimal operation is twofold. First, the set of relevant beamforming vectors is reduced to a considerably small subset of all feasible beamforming vectors. Second, the characterized set of efficient beamforming vectors is parameterized by a number of scalars which can even reduce the complexity of indicating the required beamforming vectors.

Characterizing the beamforming vectors necessary to achieve all Pareto optimal points in the MISO IFC has been carried out in several works for the case of perfect channel state information (CSI) at the transmitters. In \cite{Jorswieck2008a}, the efficient beamforming vectors are parameterized by $K(K-1)$ complex-valued parameters, where $K$ is the number of links. For the special two-user case, the efficient beamforming vectors are proven to be a linear combination of maximum ratio transmission and zero forcing transmission. The extension to a real-valued parametrization for the general $K$-user case is conducted in \cite{Shang2011,Zhang2010,Mochaourab2011} where $K(K-1)$ real-valued parameters are required to achieve all Pareto optimal points. Recently in \cite{Bjornson2011a}, parametrization of the efficient beamforming vector is provided in the multi-cell MISO setting with general linear transmit power constraints at the transmitters. For the case of MISO IFC and total power constraint at the transmitter, the number of required parameters is $2K - 1$. In \cite{Mochaourab2011a,Lindblom2011}, the beamforming vectors necessary and sufficient to achieve all Pareto optimal points are characterized and parameterized by a single real-valued parameter for the two-user MISO IFC case. If the transmitters only know the covariance matrices of the channel vectors, characterization of the Pareto optimal transmission strategies is done for the two-user MISO IFC in \cite{Lindblom2010a}. In a multi-cell MISO setting with uncertainty in CSI at the transmitters, robust Pareto optimal beamforming is obtained by robust fairness-profile optimization in \cite{Bjornson2011}. All Pareto optimal points in the performance region are achieved requiring $K-1$ real-valued parameters. In addition, a monotonic optimization algorithm is applied to achieve specific Pareto optimal points such as maximum sum and proportional fair performance points.

The presence of CSI at a transmitter is essential in order to increase the performance of the multi-antenna system through sophisticated beamforming techniques \cite{Goldsmith2003}. In practical communication systems, CSI at the transmitter is usually not perfect \cite{Vu2007}. The transmitter gains CSI either through reciprocity of the uplink and downlink channels or through a feedback link from the receiver \cite{Vu2007a}. Both techniques of acquiring CSI at the transmitter entail a delay which leads to outdated channel information and hence to a mismatch to the true channels. In the feedback model, the receiver quantizes the channel vector according to a vector codebook and sends the transmitter the channel index which is nearest to the true channel. This mechanism reduces the feedback bits required to identify the channel at the transmitter \cite{Love2008}. The size of the codebook used determines the feedback overhead and the accuracy of the selected channel vector to the true channel \cite{Jindal2006}.

With imperfect CSI at the transmitter in a multi-user MISO downlink system, robust beamforming is studied in \cite{Shenouda2008} to minimize the worst-case mean-square-error (MSE), and in \cite{Shenouda2009,Vucic2009} the problem of minimizing the transmission power subject to worst-case quality-of-service (QoS) constraints at the receivers is analyzed. Furthermore, robust transceiver design in a multi-user multiple-input multiple-output (MIMO) system is considered in \cite{Vucic2009a,Zheng2008}. In the multi-cell multiuser MISO setting with imperfect CSI at the transmitters, the problems of maximizing the worst-case weighted sum rate and minimizing the weighted sum transmission power subject to worst-case QoS constraints are studied in \cite{Tajer2011} and \cite{Shen2011a}, respectively.

In this paper, we assume that the transmitters have imperfect CSI. We adopt a deterministic uncertainty model in which the channel estimation error is bounded in an uncertainty region. The channel uncertainty region is assumed to be an ellipsoid \cite{Lorenz2005} whose geometry is known at the transmitter. Robust beamforming \cite{Li2006} takes into account the worst-case channel estimate in the uncertainty region. We derive the worst-case power gains at the receivers and formulate accordingly the worst-case achievable rates for the links. The robust rate region is the set of all worst-case jointly achievable rate tuples. In order to characterize the Pareto optimal points of the robust rate region, we show that the objective of each transmitter is a tradeoff between maximizing the intended worst-case power gain and minimizing the worst-case interference gains. This multi-objective problem is cast as a second order cone program (SOCP) which can be solved efficiently. Consequently, we characterize the beamforming vectors that are necessary to achieve all Pareto optimal points in the robust rate region. Moreover, these beamforming vectors are parameterized by $K(K-1)$ real-valued parameters taking values between zero and one. Afterwards, we analyze the spectral efficiency of the system for asymptotic values of signal-to-noise ratio (SNR). At high SNR, achieving full multiplexing gain with zero forcing transmission requires the channel estimation error to reduce linearly with SNR. If the error is independent of SNR, then single-user transmission is optimal. At low SNR, it is shown that joint maximum ratio transmission is optimal to minimize the minimum energy per bit for reliable communication.

The paper is organized as follows. In Section \ref{sec:sysmodel}, we describe the system model as well as the channel uncertainty model. The worst-case power gains at the receivers are derived and the worst-case achievable rates for the links are formulated. In Section \ref{sec:Pareto}, we characterize the beamforming vectors that are necessary to achieve Pareto optimal points in the robust rate region. As a special case, the two-user MISO IFC with spherical channel uncertainty is analyzed. In Section \ref{sec:HighLowSNR}, spectral efficiency in the high and low SNR regime are studied. In Section \ref{sec:conclusion}, we draw the conclusions.

\subsection*{Notations}
Column vectors and matrices are given in lowercase and uppercase boldface letters, respectively. $\norm{\mat{a}}$ is the Euclidean norm of $\mat{a} \in \mathbb{C}^{N}$. $\abs{b}$ is the absolute value of $b \in \mathbb{C}$. $(\cdot)^T$ and $(\cdot)^\H$ denote transpose and Hermitian transpose, respectively. The orthogonal projector onto the column space of $\mat{Z}$ is $\mat{\Pi}_{Z} := \mat{Z}(\mat{Z}^\H\mat{Z})^{-1}\mat{Z}^\H$. The orthogonal projector onto the orthogonal complement of the column space of $\mat{Z}$ is $\mat{\Pi}_{Z}^{\perp} := \bI - \mat{\Pi}_{Z}$, where $\bI$ is an identity matrix. $\mathcal{CN}(0,\mat{A})$ denotes a circularly-symmetric Gaussian complex random vector with covariance matrix $\mat{A}$. $(a)_+$ denotes $\max(a,0)$. $Re(a)$, $Im(a)$ and $\angle(a)$ denote the real part, imaginary part, and the phase of a complex number $a$.

\section{System Model}\label{sec:sysmodel}

Consider a $K$-user MISO IFC. Each transmitter has an intended receiver, and the signal from a transmitter to an unintended receiver is treated as interference. Each transmitter $k$ is equipped with $N_k$ antennas while all receivers use single antennas. The quasi-static block flat-fading channel vector from transmitter $k$ to receiver $\ell$ is denoted by $\bh_{k\ell} \in \mathbb{C}^{N_k}$. We assume that transmission consists of scalar coding followed by beamforming. The beamforming vector used by transmitter $k$ is $\mat{w}_{k} \in \mathbb{C}^{N_k}$. Each transmitter $k$ has a total power constraint of $P_k$ such that $\snorm{\mat{w}_{k}} \leq P_k$. The matched-filtered, symbol-sampled complex baseband data received at receiver $\ell$ is
\begin{equation}
y_\ell = \sum\nolimits_{k = 1}^K \bh_{k\ell}^\H \bw_k s_k + n_\ell,
\end{equation}
\noindent where $s_k \sim \mathcal{CN}(0,1)$ is the symbol transmitted by transmitter $k$. The random variables $n_k \sim \mathcal{CN}(0,\sigma^2)$ are additive Gaussian noise.

We assume that perfect CSI is present at the receivers while at the transmitters the channels are not perfectly known. A transmitter has estimates of the true channels between itself and all receivers, and we assume that the error in the channel estimation vectors is bounded in an uncertainty region. The ellipsoidal channel uncertainty region is described next.
\subsection{Channel Uncertainty Model}

Let $\ebh_{k\ell}$ be the estimate of the true channel vector $\bh_{k\ell}$ at transmitter $k$. The uncertainty in the channel estimate can be modeled using a \emph{channel estimation error vector} \cite{Bengtsson2001}
\begin{equation}
\mat{\delta}_{k\ell} = \bh_{k\ell} - \ebh_{k\ell}, \quad \mat{\delta}_{k\ell} \in \mathcal{E}_{k\ell},
\end{equation}
\noindent where $\mat{\delta}_{k\ell}$ is assumed to be bounded in an \emph{ellipsoidal uncertainty region} $\mathcal{E}_{k\ell}$ defined as \cite{Lorenz2005}\footnote{In \cite{Lorenz2005}, the ellipsoidal uncertainty region is defined as $\mathcal{E}_{k\ell}=\{\mat{A}_{k\ell} \mat{\delta} + \mat{c}_{k\ell}:\norm{\mat{\delta}} \leq \epsilon_{k\ell}\}$, where $\mat{c}_{k\ell}$ defines the center of the ellipsoid $\mathcal{E}_{k\ell}$. Here, we assume $\mat{c}_{k\ell} = \mat{0}$ without loss of generality.}
\begin{equation}\label{eq:ellipsoid}
\mathcal{E}_{k\ell} = \{\mat{A}_{k\ell} \mat{\delta}:\norm{\mat{\delta}} \leq \epsilon_{k\ell}\}.
\end{equation}
\noindent In \eqref{eq:ellipsoid}, $\mat{A}_{k\ell} \in \mathbb{C}^{N_k \times N_k}$ determines the shape of the ellipsoid. We assume that $\mathcal{E}_{k\ell}$ is full rank for all $k,\ell = 1,\ldots,K,$ and has a maximum radius of $\epsilon_{k\ell}$, i.e. the largest singular value of $\mat{A}_{k\ell}$ is constrained to one. Generally, the shape and size of the ellipsoid in \eqref{eq:ellipsoid} should be specified according to the error probability in channel estimation for the system under consideration. For example, if channel vector estimation is done at a transmitter using training-sequences from the corresponding receiver, the size of the uncertainty region can be chosen as a scaled version of the channel estimation mean square error (MSE) \cite{Zheng2008,Bjornson2011}.

The elliptical uncertainty model described in \eqref{eq:ellipsoid} is more general and encompasses the spherical uncertainty model adopted in \cite{Vorobyov2003}. The spherical uncertainty region is defined as
\begin{equation}\label{eq:sphere}
\mathcal{D}_{k\ell} = \{\mat{\delta}:\norm{\mat{\delta}} \leq \epsilon_{k\ell}\}.
\end{equation}
\noindent Specifically, $\mathcal{E}_{k\ell}$ in \eqref{eq:ellipsoid} is a sphere when $\mat{A}_{k\ell} = \mat{I}$.

A transmitter $k$ knows the channel estimates $\ebh_{k\ell}$ and the associated uncertainty region $\mathcal{E}_{k\ell}$. Robust transmission requires the choice of beamforming vectors to be robust to channel estimation errors. Thus, a transmitter has to consider worst-case achievable rate at its intended receiver for secure communication.

\subsection{Worst-Case Achievable Rates}
The worst-case achievable rate for link $\ell$ with single-user decoding is
\begin{equation}\label{eq:rate}
R_\ell(\bw_1,\ldots,\bw_K) = \ld\pp{1 + \frac{x^2_{\ell \ell}(\bw_\ell)}{\sigma^2 + \sum_{k\neq \ell} x_{k\ell}^2(\bw_k)}},
\end{equation}
\noindent where $x^2_{k \ell}(\bw_k)$ is the worst-case signal power from transmitter $k$ to receiver $\ell$. The worst-case intended power gain $x^2_{\ell\ell}(\bw_\ell)$ is the least power gain achievable within the uncertainty set $\mathcal{E}_{\ell\ell}$. That is,
\begin{subequations}\label{eq:intendedgain}
\begin{align}
    x_{\ell\ell}(\bw_\ell) &= \min_{\mat{\delta}_{\ell\ell} \in \mathcal{E}_{\ell\ell}} \abs{\bh_{\ell\ell}^H \bw_\ell} \quad \text{with } \bh_{\ell\ell} = \ebh_{\ell\ell} + \mat{\delta}_{\ell\ell} \\
    &= \min_{\mat{\delta}_{\ell\ell} \in \mathcal{E}_{\ell\ell}} \abs{\ebh_{\ell\ell}^H \bw_\ell + \mat{\delta}_{\ell\ell}^H \bw_\ell}\\
    & = \min_{\norm{\mat{\delta}}\leq \epsilon_{\ell\ell}} \abs{\ebh_{\ell\ell}^H \bw_\ell + \mat{\delta}^H \mat{A}_{\ell\ell}^H \bw_\ell}.
\end{align}
\end{subequations}
The error vector which minimizes the intended power gain in \eqref{eq:intendedgain} can be calculated following similar steps as in \cite{Vorobyov2003}. Using the triangle inequality \cite[3.2.5]{Abramowitz1972}, we have
\begin{equation}\label{eq:triangle}
    \abs{\ebh_{\ell\ell}^H \bw_\ell + \mat{\delta}^H \mat{A}_{\ell\ell}^H \bw_\ell} \geq \begin{cases} 0 & \text{if } \abs{\ebh_{\ell\ell}^H \bw_\ell} < \abs{\mat{\delta}^H \mat{A}_{\ell\ell}^H \bw_\ell} \\ \abs{\ebh_{\ell\ell}^H \bw_\ell} - \abs{\mat{\delta}^H \mat{A}_{\ell\ell}^H \bw_\ell} & \text{otherwise}
    \end{cases},
\end{equation}
\noindent for all $\mat{\delta}$ such that $\norm{\mat{\delta}}\leq \epsilon_{\ell\ell}$. In \eqref{eq:triangle}, the lower bound of zero for $\abs{\ebh_{\ell\ell}^H \bw_\ell} < \abs{\mat{\delta}^H \mat{A}_{\ell\ell}^H \bw_\ell}$ is achievable by a negative linear scaling of $\mat{\delta}$. Otherwise, if $\abs{\ebh_{\ell\ell}^H \bw_\ell} \geq \abs{\mat{\delta}^H \mat{A}_{\ell\ell}^H \bw_\ell}$, then the error vector which achieves the lower bound in \eqref{eq:triangle} is
\begin{equation}\label{eq:triangle1}
    \mat{\delta} = - \epsilon_{\ell\ell} \frac{\mat{A}_{\ell\ell}^H \bw_\ell}{\norm{\mat{A}_{\ell\ell}^H \bw_\ell}} e^{j \angle(\ebh_{\ell\ell}^H \bw_\ell)}.
\end{equation}
\noindent Substituting \eqref{eq:triangle1} in the RHS of \eqref{eq:triangle}, the square root of the worst-case intended power gain in \eqref{eq:intendedgain} reduces to
\begin{equation}\label{eq:dir_gain}
     x_{\ell\ell}(\bw_\ell) = \pp{\abs{\ebh_{\ell\ell}^H \bw_\ell} - \norm{\mat{A}^H_{\ell\ell} \bw_\ell}\epsilon_{\ell\ell}}_+.
\end{equation}
The square root of the worst-case interference power gain from transmitter $k$ to receiver $\ell, k \neq \ell$, is
\begin{subequations}\label{eq:interferencegain}
\begin{align}
    x_{k\ell}(\bw_k) &= \max_{\mat{\delta}_{k\ell} \in \mathcal{E}_{k\ell}} \abs{\bh_{k\ell}^H \bw_k} \quad \text{with } \bh_{k\ell} = \ebh_{k\ell} + \mat{\delta}_{k\ell} \\
    &= \max_{\mat{\delta}_{k\ell} \in \mathcal{E}_{k\ell}} \abs{\ebh_{k\ell}^H \bw_k + \mat{\delta}_{k\ell}^H \bw_k}\\
    & = \max_{\norm{\mat{\delta}}\leq \epsilon_{k\ell}} \abs{\ebh_{k\ell}^H \bw_k + \mat{\delta}^H \mat{A}_{k\ell}^H \bw_k}.
\end{align}
\end{subequations}
Using the triangle inequality, we have
\begin{equation}\label{eq:triangle_int}
    \abs{\ebh_{k\ell}^H \bw_k + \mat{\delta}^H \mat{A}_{k\ell}^H \bw_k} \leq \abs{\ebh_{k\ell}^H \bw_k} + \abs{\mat{\delta}^H \mat{A}_{k\ell}^H \bw_k},
\end{equation}
\noindent and the upper bound in \eqref{eq:triangle_int} is achieved by
\begin{equation}
    \mat{\delta} = \epsilon_{k\ell} \frac{\mat{A}_{k\ell}^H \bw_k}{\norm{\mat{A}_{k\ell}^H \bw_k}} e^{j \angle ({\ebh}_{k\ell}^H \bw_k)}.
\end{equation}
\noindent Hence, the worst-case interference gain is
\begin{equation}\label{eq:int_gain}
     x_{k\ell}(\bw_k) = \abs{\ebh_{k\ell}^H \bw_k} + \norm{\mat{A}^H \bw_k}\epsilon_{k\ell}.
\end{equation}

Note that the worst-case measures in \eqref{eq:dir_gain} and \eqref{eq:int_gain} are the main terms in the worst-case rate expression in \eqref{eq:rate}. Also, the channel estimation errors of all channel vectors are independent such that the worst-case interference and intended power gains may occur simultaneously.
\section{Pareto Optimal Beamforming}\label{sec:Pareto}

The $K$-dimensional \emph{robust rate region} consists of all jointly achievable worst-case rate tuples defined as
\begin{equation}\label{eq:RR}
\mathcal{R} := \{\pp{R_1(\bw_1,\ldots,\bw_K),\ldots,R_K(\bw_1,\ldots,\bw_K)}: \snorm{\bw_k}\leq P_k, k = 1,\ldots,K\},
\end{equation}
\noindent where $R_k(\bw_1,\ldots,\bw_K)$ is defined in \eqref{eq:rate}. The region $\mathcal{R}$ in \eqref{eq:RR} is not necessarily a convex set\footnote{Note that any point on the convex hull of $\mathcal{R}$ can be reached by time-sharing. In this work, we do not consider time-sharing techniques to convexify the rate region $\mathcal{R}$.} \cite{Charafeddine2011}. This is mainly due to the interference coupling present between the links. The outer boundary of the rate region comprises efficient operating points. These points are called Pareto optimal.

\begin{definition}\label{def:Pareto}
A rate tuple $(R_1,...,R_K) \in \mathcal{R}$ is \emph{Pareto optimal} if there is no other tuple $({R'}_1,...,{R'}_K) \in \mathcal{R}$ such that $({R'}_1,...,{R'}_K) \geq (R_1,...,R_K)$, where the inequality is component-wise and strict for at least one component. The set of all Pareto optimal operating points constitutes the \emph{Pareto boundary} of $\mathcal{R}$.
\end{definition}

In order to achieve a Pareto optimal point, each beamforming vector must include a tradeoff between the maximization of the intended power gain and the minimization of the interference power gains. Next, we will consider the special two-user MISO IFC case and compare Pareto efficient beamforming with perfect and imperfect CSI at the transmitters. Also as a special case, the channel uncertainty regions are considered to be spherical. Later, the general $K$-user MISO IFC with elliptical channel uncertainty region will be restored.

\subsection{Two-User Case with Spherical Channel Uncertainty}
In the two-user case, the rate region in \eqref{eq:RR} is two-dimensional rewritten as
\begin{equation}\label{eq:RR2}
\mathcal{R} := \{\pp{R_1(\bw_1,\bw_2),R_2(\bw_1,\bw_2)}: \snorm{\bw_k}\leq P_k, k = 1,2\},
\end{equation}
\noindent where the achievable rate for link $\ell$ is
\begin{equation}\label{eq:rate1}
R_\ell(\bw_1,\bw_2) = \ld\pp{1 + \frac{\sabs{\bh_{\ell\ell}^H \bw_\ell}}{\sigma^2 + \sabs{\bh_{k\ell}^H \bw_k}}}, \quad k\neq \ell.
\end{equation}
If the transmitter knows the channel vectors perfectly, the beamforming vectors that are necessary to achieve all Pareto optimal points of the two-user rate region in \eqref{eq:RR2} are \cite{Jorswieck2008a}
\begin{equation}\label{eq:ParOptBeam_perfect}
    \bw_k(\alpha_k) =  \sqrt{P_k} \pp{\sqrt{\alpha_k} \frac{\Pi_{\bh_{k\ell}} \bh_{kk}}{\norm{\Pi_{\bh_{k\ell}} \bh_{kk}}} + \sqrt{1 - \alpha_k} \frac{\Pi^\perp_{\bh_{k\ell}} \bh_{kk}}{\norm{\Pi^\perp_{\bh_{k\ell}} \bh_{kk}}}}, \quad k\neq \ell,
\end{equation}
\noindent where $\alpha_k \in [0,\alpha_k^\mrt]$, with $\alpha_k^\mrt = \norm{\Pi_{\bh_{k\ell}}\bh_{kk}}^2/\norm{\bh_{kk}}^2$. In \eqref{eq:ParOptBeam_perfect}, the beamforming vectors that achieve Pareto optimal points are a combination of two orthogonal unit norm vectors. Zero forcing (ZF) to the unintended receiver corresponds to $\alpha_k = 0$, and maximum ratio transmission (MRT) to $\alpha_k = \alpha_k^\mrt$.

If the transmitters have imperfect CSI and the channel vector uncertainty region is spherical as defined in \eqref{eq:sphere}, then the worst-case achievable rate of link $\ell$ is
\begin{equation}\label{eq:rate2}
R_\ell(\bw_1,\bw_2) = \ld\pp{1 + \frac{\pp{\abs{\ebh_{\ell\ell}^H \bw_\ell} - \norm{\bw_\ell}\epsilon_{\ell\ell}}_+^2}{\sigma^2 + \pp{\abs{\ebh_{k\ell}^H \bw_k} + \norm{\bw_k}\epsilon_{k\ell}}^2}}, \quad k\neq \ell,
\end{equation}
\noindent where the intended and interference power gains are from \eqref{eq:dir_gain} and \eqref{eq:int_gain}, respectively, with $\mat{A}_{k\ell} = \mat{I}$. According to the rate expression in \eqref{eq:rate2}, worst-case signal power and interference powers include additive terms influenced only by the norm of the beamforming vectors. It is thus expected that robust Pareto optimal beamforming includes additionally varying transmission power. The Pareto boundary of the rate region $\mathcal{R}$ in \eqref{eq:RR2} is achieved by the beamforming vectors \cite{Mochaourab2011b}
\begin{equation}\label{eq:ParOptBeam_imperfect}
    \bw_k(\xi_k,\beta_k) =  \sqrt{ \xi_k P_k} \pp{\sqrt{\beta_k} \frac{\Pi_{\ebh_{k\ell}} \ebh_{kk}}{\norm{\Pi_{\ebh_{k\ell}} \ebh_{kk}}} + \sqrt{1 - \beta_k} \frac{\Pi^\perp_{\ebh_{k\ell}} \ebh_{kk}}{\norm{\Pi^\perp_{\ebh_{k\ell}} \ebh_{kk}}}}, \quad k \neq \ell,
\end{equation}
\noindent where $\xi_k \in [0,1]$, and $\beta_k \in [0,\beta_k^\rmrt]$, with $\beta_k^\rmrt = \norm{\Pi_{\ebh_{k\ell}}\ebh_{kk}}^2/\norm{\ebh_{kk}}^2$. The parametrization in \eqref{eq:ParOptBeam_imperfect} requires two parameters per transmitter. One parameter is to alter the direction of the beamforming vector and one parameter scales the transmission power. Interestingly, the structure of the efficient beamforming vectors with spherical channel uncertainty in \eqref{eq:ParOptBeam_imperfect} is similar to the case of perfect CSI at the transmitters in \eqref{eq:ParOptBeam_perfect}. In \eqref{eq:ParOptBeam_imperfect}, the channel estimates replace the true channel vectors in \eqref{eq:ParOptBeam_perfect}. The beamforming vectors corresponding to MRT and ZF represent extreme strategies which have the objective of either maximizing the power at the intended receiver or minimizing the interference power gain. Robust MRT is calculated as
\begin{equation}\label{eq:R-MRT}
    \bw_k^\rmrt = \argmax_{\norm{\bw_k} \leq \sqrt{P_k}} \quad \abs{\ebh_{kk}^H \bw_k} - \norm{\bw_k} \epsilon_{kk} = \sqrt{P_k} \frac{\ebh_{kk}}{\norm{\ebh_{kk}}}.
\end{equation}
\noindent Thus, to maximize the power gain at the intended receiver in the worst-case of spherical uncertainty, the transmitter chooses full power transmission in the direction of the estimated channel. In \cite{Wiesel2007}, it is shown that robust MRT corresponds to the channel vector estimate whenever the uncertainty region is symmetric. Robust ZF, however is achieved by allocating zero power. This is observed in the denominator of \eqref{eq:rate2} where the interference gain can only be zero for $\norm{\bw_k} = 0$.

Next, we address the general $K$-user case with general ellipsoidal uncertainty regions defined in \eqref{eq:ellipsoid}.

\subsection{$K$-User Case with Ellipsoidal Channel Uncertainty}\label{sec:K-User}

As in the two-user case, Pareto optimal beamforming in the $K$-user case requires a tradeoff between maximizing intended power gain and minimizing interference gains. This is concluded by observing that the worst-case achievable rate in \eqref{eq:rate} is monotonically increasing with the direct power gain $x^2_{\ell\ell}(\bw_\ell)$ for fixed interference power $x^2_{k\ell}(\bw_k), k \neq \ell$. In addition, the worst-case achievable rate is monotonically decreasing with the interference power gain for fixed intended power gain. In other words, increasing the power gain at the intended receiver increases its achievable rate and reducing the interference at unintended receivers increases the rates of the other links. For a transmitter $k$, the tradeoff between the two objectives of increasing  $x^2_{kk}(\bw_k)$ and reducing $x^2_{k\ell}(\bw_k), k \neq \ell,$ can be cast as a multi-objective optimization problem \cite{Marler2004,Zhang2007}
\begin{equation}\label{eq:multobj}
\bw_k = \argmax~ \mat{g}_k(\bw_k) \quad \text{ s.t. } \snorm{\bw_k} \leq P_k,
\end{equation}
\noindent where the multi-objective function $\mat{g}_k: \mathbb{C}^{N_k}\rightarrow \mathbb{R}^K$ is defined as
\begin{equation}
{g}_{k\ell}(\bw_k) =
\left\{
  \begin{array}{ll}
    -x^2_{k\ell}(\bw_k), & k \neq \ell; \\
    +x^2_{kk}(\bw_k), & k = \ell.
  \end{array}
\right.
\end{equation}
There are several methods to solve the problem in \eqref{eq:multobj}. One technique is the weighted sum method \cite{LeyfferSpring2009}
\begin{equation}\label{eq:weightedMOP}
\bw_k(\mat{\gamma}_k) = \argmax~ \gamma_{kk} {x}^2_{kk}(\bw_k) - \sum_{\ell \neq k} \gamma_{k\ell}{x}^2_{k\ell}(\bw_k) \quad \text{ s.t. } \snorm{\bw_k} \leq P_k,
\end{equation}
\noindent where $\gamma_{k\ell}$ are nonnegative weights such that $\sum_{\ell = 1}^K \gamma_{k\ell} = 1$. In the $K$-user MISO IFC with perfect CSI at the transmitters, the beamforming vectors necessary to achieve all Pareto optimal points are found in \cite{Mochaourab2011} by the optimization problem in \eqref{eq:weightedMOP}. For the perfect CSI case, the objective in \eqref{eq:weightedMOP} is a sum of Hermitian forms and hence the optimization is an eigenvalue problem.

Another method for solving \eqref{eq:multobj} is by maximizing a single objective in $\mat{g}_k(\bw_k)$ while setting goals for the other objectives \cite{LeyfferSpring2009}. This method is called the boundary intersection approach \cite{Zhang2007} and casts the problem in \eqref{eq:multobj} as
\begin{subequations}\label{eq:goalProgramm}
\begin{align}
\bw_k(\mat{\lambda}_k) = \argmax & \quad x^2_{kk}(\bw_k)\\
\text{s.t.} & \quad x^2_{k\ell} \leq \lambda_{k\ell} \Gamma_{k\ell}, ~ \text{for all } \ell \neq k,\\
& \quad \snorm{\bw_k} \leq {P_k},
\end{align}
\end{subequations}
\noindent with $\lambda_{k\ell} \in [0,1]$ and $\Gamma_{k\ell}$ is a fixed value corresponding to an upper limit on the interference level generated by transmitter $k$. The formulation in \eqref{eq:goalProgramm} is common for optimization problems in cognitive radio networks where a secondary user maximizes his intended power gain or achievable rate subject to interference temperature constraints at the primary receivers \cite{Zhang2009,Zheng2010,Wang2011}. Also, in \cite{Zhang2010} the necessary beamforming vectors to achieve all Pareto optimal points in the MISO IFC with perfect CSI are found by maximizing the achievable rate while setting interference constraints at the unintended receivers.

In order to solve the multi-objective optimization problem in \eqref{eq:multobj} with channel uncertainty at the transmitters, we use the method in \eqref{eq:goalProgramm} and show that it leads to a convex optimization problem \cite{Boyd2004}. By taking the square root of the objective function and the constraints and substituting the power gains from \eqref{eq:dir_gain} and \eqref{eq:int_gain} in \eqref{eq:goalProgramm} we get
\begin{subequations}\label{eq:ParOptBeam0}
\begin{align}
\bw_k(\mat{\lambda}_k) = \argmax & \quad \abs{\ebh_{k k}^\H \bw_k} - \norm{\bA_{kk}^\H \bw_k}\epsilon_{kk}\\
\text{s.t.} & \quad \abs{\ebh_{k \ell}^\H \bw_k} + \norm{\bA_{k \ell}^\H \bw_k}\epsilon_{k\ell} \leq \sqrt{\lambda_{k \ell} \Gamma_{k \ell}},\\ \nonumber & \quad \ell = 1,\ldots,K, ~ \ell \neq k,\\
& \quad \norm{\bw_k} \leq \sqrt{P_k}.
\end{align}
\end{subequations}
\noindent Since the objective function does not depend on the phase of the beamforming vector $\bw_k$ \cite{Vorobyov2003}, the problem in \eqref{eq:ParOptBeam0} can be equivalently written as
\begin{subequations}\label{eq:ParOptBeam}
\begin{align}
\bw_k(\mat{\lambda}_k) = \argmax & \quad \text{Re}\br{\ebh_{k k}^\H \bw_k} - \norm{\bA_{kk}^\H \bw_k}\epsilon_{kk}\\
\text{s.t.} & \quad \abs{\ebh_{k \ell}^\H \bw_k} + \norm{\bA_{k \ell}^\H \bw_k}\epsilon_{k\ell} \leq \sqrt{\lambda_{k \ell} \Gamma_{k \ell}},\\ \nonumber
& \quad \ell = 1,\ldots,K, ~ \ell \neq k,\\
& \quad \text{Im}\br{\ebh_{k k}^\H \bw_k}=0,\\
& \quad \norm{\bw_k} \leq \sqrt{P_k},
\end{align}
\end{subequations}
\noindent where $\lambda_{k\ell} \in [0,1]$. The problem in \eqref{eq:ParOptBeam} is a second order cone program (SOCP) \cite{Lobo1998} which can be efficiently solved by interior point methods \cite{Boyd2004} (An optimization package that solves SOCPs is SeDuMi \cite{Sturm1999}.)

In the optimization problem \eqref{eq:ParOptBeam}, the parameter $\lambda_{k\ell}$ determines the interference level that transmitter $k$ is allowed to generate at an unintended receiver $\ell$. For $\lambda_{k\ell} = 1$, the highest interference level is allowed. For the constraints $\Gamma_{k \ell}$ with $\lambda_{k\ell}=1$ to be tight, $\Gamma_{k \ell}$ should be the interference level at receiver $\ell$ when transmitter $k$ performs MRT, i.e.,
\begin{equation}
\Gamma_{k\ell} = x^2_{k\ell}(\bw_k^\rmrt) = \pp{\abs{\ebh_{k \ell}^\H \bw_k^\rmrt} + \norm{\bA_{k \ell}^\H \bw_k^\rmrt}\epsilon_{k\ell}}^2,
\end{equation}
\noindent with MRT beamforming obtained from
\begin{subequations}\label{eq:MRT_rob}
\begin{align}
\bw_k^\rmrt = \argmax & \quad \text{Re}\br{\ebh_{k k}^\H \bw_k} - \norm{\bA_{kk}^\H \bw_k} \epsilon_{kk}\\
\text{s.t.} & \quad \text{Im}\br{\ebh_{k k}^\H \bw_k}=0,\\
& \quad \norm{\bw_k} \leq \sqrt{P_k}.
\end{align}
\end{subequations}
\noindent Another method for calculating robust MRT is provided in \cite{Wiesel2007} which uses a one dimensional line search.
\begin{theorem}\label{thm:ParetoOpt}
All Pareto optimal points of the robust rate region $\mathcal{R}$ in \eqref{eq:RR} can be achieved by the beamforming vectors $\bw_k(\mat{\lambda}_k)$ from \eqref{eq:ParOptBeam} with $\mat{\lambda}_k \in [0,1]^{K-1}$ for $k = 1,...,K$.
\end{theorem}
\begin{proof}
The proof is by contradiction. Assume a beamforming vector $\bw_k$ which is not a solution of the problem in \eqref{eq:ParOptBeam} achieves a Pareto optimal point $(R_1,\ldots,R_K)$. Then, it is possible to find a beamforming vector $\bw'_k$ from \eqref{eq:ParOptBeam} in which the intended power gain increases at receiver $k$, i.e. $x^2_{kk}(\bw'_k) > x^2_{kk}(\bw_k)$, without affecting the interference gain to all other receivers, i.e. $x^2_{k\ell}(\bw'_k) = x^2_{k\ell}(\bw_k)$ for all $\ell \neq k$. In this case the achievable rate of link $k$ is increased without affecting the rates of the other receivers. According to Definition \ref{def:Pareto}, $(R_1,\ldots,R_K)$ would not be Pareto optimal which is a contradiction to the original assumption. Hence, any beamforming vector not in the solution set of \eqref{eq:ParOptBeam} does not achieve a Pareto optimal point in the robust rate region $\mathcal{R}$. $\hfill\Box$
\end{proof}

The Pareto boundary of the robust rate region can be also obtained by the framework provided in \cite{Bjornson2011}. The approach in \cite{Bjornson2011} uses a robust fairness-profile optimization to calculate a Pareto optimal point which entails solving a set of convex feasibility problems. Accordingly, only points on the Pareto boundary are delivered requiring $K-1$ real-valued parameters. In comparison to the parametrization in \cite{Bjornson2011}, our result in Theorem \ref{thm:ParetoOpt} specifies the beamforming vectors for each transmitter that are necessary to achieve all Pareto optimal points in the robust rate region. The efficient beamforming vectors for each transmitter are parameterized by $K-1$ real-valued parameters each between zero and one. The parametrization of the efficient beamforming vectors can be utilized for designing efficient low complexity distributed resource allocation schemes which require low signaling overhead between the transmitters as is done in \cite{Ho2008,Mochaourab2010}.

\begin{figure}[h]
\centering
[Figure 1 about here]
\end{figure}

In \figurename~\ref{fig:rateregion}, a three-user robust rate region is plotted. The total power constraints at the transmitters are set to one and the SNR is $1/\sigma^2 = 0$ dB. The number of antennas at each transmitter is three. The channel vector estimates $\ebh_{k\ell}$ are independent and identically distributed as $\ebh_{k\ell} \sim \mathcal{CN}(0,\mat{I})$. Each ellipsoidal uncertainty region is generated as follows: $N_k$ vectors are generated according to complex normal distribution with zero mean and covariance matrix $\mat{I}$. $\mat{A}_{k\ell}$ is constructed by concatenating the generated vectors and normalizing the constructed matrix such that the largest singular value of $\mat{A}_{k\ell}$ is one. The errors are equally chosen as $\epsilon_{k\ell} = 0.5$ for $k,\ell = 1,2,3$. For each transmitter $k$, we generate the efficient beamforming vectors from \eqref{eq:ParOptBeam} with the parameters $\lambda_{k\ell}$ uniformly sampled in a $0.05$ step-length between zero and one. A set of rate tuples is calculated from the generated beamforming vectors of the transmitters. Within this set, a subset corresponds to Pareto optimal points of the robust rate region. In \figurename~\ref{fig:rateregion}, only the Pareto optimal points are plotted by utilizing the MATLAB code in \cite{Cao2007}.

\section{Spectral Efficiency at High and Low SNR}\label{sec:HighLowSNR}
In this section, optimal beamforming is studied for asymptotic values of SNR. Moreover, we analyze the effects of imperfect CSI at the transmitters on the performance of the system in comparison to the case of perfect CSI.

\subsection{Efficiency at High SNR}
The quantitative performance is analyzed using the high-SNR offset concept in \cite[Section $\mathrm{II}$]{Lozano2005}. We define the SNR as $\rho=1/\sigma^2$ and the maximum sum rate as a function of SNR as $R^\tsum(\rho)$. The high-SNR slope is
\begin{equation}\label{eq:Sinfty}
S_{\infty} = \lim\limits_{\rho \rightarrow \infty}\frac{R^\tsum(\rho)}{\log_2(\rho)},
\end{equation}
which corresponds to the multiplexing gain, i.e. the slope of the maximum sum rate curve at high SNR. The maximum sum rate is
\begin{subequations}\label{eq:SumRate}
\begin{align}
R^\tsum(\rho) &= \max_{\bw_1,\ldots,\bw_K} \sum_{\ell=1}^K \log_2\pp{1+\frac{\rho x_{\ell\ell}^2(\bw_\ell)}{1+\rho \sum_{k \neq \ell} x_{k\ell}^2(\bw_k)}}\\
& = \max_{\bw_1,\ldots,\bw_K} \sum_{\ell=1}^K \log_2\pp{1+\frac{\rho \pp{\abs{\ebh_{\ell\ell}^H \bw_\ell} - \norm{\mat{A}^H_{\ell\ell} \bw_\ell}\epsilon_{\ell\ell}}_+^2}{1+\rho \sum\limits_{k \neq \ell} \pp{\abs{\ebh_{k\ell}^H \bw_k} + \norm{\mat{A}^H_{k\ell} \bw_k}\epsilon_{k\ell}}^2}},
\end{align}
\end{subequations}
\noindent where $x_{\ell\ell}(\bw_\ell)$ and $x_{k\ell}(\bw_k)$ are from \eqref{eq:dir_gain} and \eqref{eq:int_gain} respectively. From \eqref{eq:Sinfty}, the high-SNR slope is
\begin{equation}\label{eq:Sinfty2}
S_{\infty} = \lim\limits_{\rho \rightarrow \infty}\frac{\max\limits_{\bw_1,\ldots,\bw_K} \sum\limits_{\ell=1}^K \log_2\pp{1+\frac{\rho \pp{\abs{\ebh_{\ell\ell}^H \bw_\ell} - \norm{\mat{A}^H_{\ell\ell} \bw_\ell}\epsilon_{\ell\ell}}_+^2}{1+\rho \sum\limits_{k \neq \ell} \pp{\abs{\ebh_{k\ell}^H \bw_k} + \norm{\mat{A}^H_{k\ell} \bw_k}\epsilon_{k\ell}}^2}}}{\log_2(\rho)},
\end{equation}
which is maximized when the interference is nulled at all receivers. Since $\mat{A}_{k\ell}$ is full rank, the worst-case interference gains cannot be nulled unless the transmitters switch their transmission off. In case the error $\epsilon_{k\ell}$ does not depend on the SNR, the interference gains as well as the intended power gain in the achievable rate of a link $\ell$ scale linearly with $\rho$. Hence, the high-SNR slope in \eqref{eq:Sinfty2} is zero if more than one link operate simultaneously. Therefore, in the high-SNR regime single-user transmission is optimal achieving the largest high-SNR slope of $S_\infty = 1$. The maximum sum rate is then
\begin{equation}\label{eq:SumRateImpCSI2}
    R^\tsum(\rho) = \log_2\pp{1+\rho \max_{\ell = 1,\ldots,K}\pp{\abs{\ebh_{\ell\ell}^H \bw^\mrt_\ell} - \norm{\mat{A}^H_{\ell\ell} \bw^\mrt_\ell}\epsilon_{\ell\ell}}_+^2},
\end{equation}
\noindent where only one user operates using MRT and full power transmission. Note that the condition that determines the dominant user does not only depend on the channel gains but also on the amount of uncertainty present at the transmitter.

In \cite{Jindal2006} it is shown that the channel estimation error has to scale linearly with the inverse SNR in order to achieve maximum multiplexing gain with ZF in a MISO broadcast channel. That is, $\epsilon^2_{k\ell}(\rho) \propto \frac{1}{\rho}$ should hold for all $k$ and $\ell$. We assume that the estimation error has the following dependence on SNR:
\begin{equation}\label{eq:errorScaling}
\epsilon_{k\ell}(\rho) = \frac{a_{k\ell}}{\sqrt{\rho}},
\end{equation}
\noindent where $a_{k\ell}$ is a constant. ZF beamforming according to the channel estimates is
\begin{equation}\label{eq:ZF}
\bw^\zf_k = \frac{\Pi_{\mat{Z}_k}^\perp \ebh_{kk}}{\norm{\Pi_{\mat{Z}_k}^\perp \ebh_{kk}}},
\end{equation}
\noindent where
\begin{equation}
\mat{Z}_k = \left[\ebh_{k1},\ldots,\ebh_{kk-1},\ebh_{kk+1},\ldots,\ebh_{kK}\right].
\end{equation}
\noindent Note that if the number of antennas at a transmitter $k$ is strictly less than the number of receivers $K$, then $\bw^\zf_k = \mat{0}$. Assuming $N_k\geq K$ for all $k=1,\ldots,K$, the maximum sum rate with ZF is
\begin{subequations}\label{eq:SumRate2}
\begin{align}
    R^\tsum(\rho)&= \sum_{\ell=1}^K \log_2\pp{1+\frac{\rho \pp{\abs{\ebh_{\ell\ell}^H \bw^\zf_\ell} - \norm{\mat{A}^H_{\ell\ell} \bw^\zf_\ell}\epsilon_{\ell\ell}(\rho)}_+^2}{1+\rho \sum\limits_{k \neq \ell} \pp{\underbrace{\abs{\ebh_{k\ell}^H \bw^\zf_k}}_{=0} + \norm{\mat{A}^H_{k\ell} \bw^\zf_k}\epsilon_{k\ell}(\rho)}^2}}\\
    &= \sum_{\ell=1}^K \log_2\pp{1+\frac{\pp{\sqrt{\rho}\abs{\ebh_{\ell\ell}^H \bw^\zf_\ell} - \sqrt{\rho} \norm{\mat{A}^H_{\ell\ell} \bw^\zf_\ell}\frac{a_{\ell\ell}}{\sqrt{\rho}}}_+^2}{1+ \sum_{k \neq \ell} \pp{ \sqrt{\rho}\norm{\mat{A}^H_{k\ell} \bw^\zf_k}\frac{a_{k\ell}}{\sqrt{\rho}}}^2}}\\
    &= \sum_{\ell=1}^K \log_2\pp{1+\frac{\pp{\sqrt{\rho}\abs{\ebh_{\ell\ell}^H \bw^\zf_\ell} - \norm{\mat{A}^H_{\ell\ell} \bw^\zf_\ell}{a_{\ell\ell}}}_+^2}{1+ \sum_{k \neq \ell} \pp{\norm{\mat{A}^H_{k\ell} \bw^\zf_k}{a_{k\ell}}}^2}}.
\end{align}
\end{subequations}
\noindent The high-SNR slope from \eqref{eq:Sinfty} with $N_k \geq K$ for all $k = 1,\ldots,K$ is
\begin{equation}\label{eq:Sinfty3}
S_{\infty} = \lim\limits_{\rho \rightarrow \infty}\frac{\sum_{\ell=1}^K \log_2\pp{1+\frac{\pp{\sqrt{\rho}\abs{\ebh_{\ell\ell}^H \bw^\zf_\ell} - \norm{\mat{A}^H_{\ell\ell} \bw^\zf_\ell}{a_{\ell\ell}}}_+^2}{1+ \sum_{k \neq \ell} \pp{\norm{\mat{A}^H_{k\ell} \bw^\zf_k}{a_{k\ell}}}^2}}}{\log_2(\rho)} = K.
\end{equation}
\noindent Hence, the maximum multiplexing gain of $K$ is achieved as with perfect CSI. Note, that we assumed that $N_k \geq K$ for all $k$. If there exists a transmitter $k$ such that $N_k < K$, then transmitter $k$ cannot perform ZF to $K$ receivers simultaneously. In this case, the highest multiplexing gain would be the maximum number of transmitters $m^*$ that have more antennas than $m^*$. The optimization problem to calculate the multiplexing gain is
\begin{subequations}\label{eq:multGain}
\begin{align}
m^* = \max & \quad m\\
\text{s.t.} & \quad m = \sum^K_{k=1} \eta_k,  \quad \eta_k =
\left\{
  \begin{array}{ll}
    1, & N_k \geq m; \\
    0, & N_k < m.
  \end{array}
\right.
\end{align}
\end{subequations}
\noindent Finding $m^*$ has to be performed iteratively. In Algorithm \ref{alg:multgain}, we provide a method to calculate $m^*$. First, the vector containing the number of antennas $[N_1,\ldots,N_K]$ is sorted in a weakly decreasing order to $[\tilde{N}_1,\ldots,\tilde{N}_K]$ with $\tilde{N}_i \geq \tilde{N}_{i+1}$. The multiplexing gain $m^*$ is found as the largest $k$ such that $\tilde{N}_k \geq k$. ZF transmission with $\tilde{N}_{m^* + 1}$ antennas reduces the multiplexing gain to $\tilde{N}_{m^* + 1} < \tilde{N}_{m^*} = m^*$. Consequently, in order to achieve the maximum multiplexing gain, only the links corresponding to the first $m^*$ entries in $[\tilde{N}_1,\ldots,\tilde{N}_K]$ should operate and perform ZF according to the channel estimates.

\begin{figure}[h]
\centering
[Algorithm 1 about here]
\end{figure}

\begin{figure}[h]
\centering
[Figure 2 about here]
\end{figure}

In \figurename~\ref{fig:sumrate}, the maximum sum rate is plotted for two links. Each transmitter has three antennas and the total power constraints at the transmitters are $P_k = 1, k=1,2$. The channel vector estimates and the uncertainty regions are generated as stated in Section \ref{sec:K-User} in the description of \figurename~\ref{fig:rateregion}. The error $\epsilon_{kj}$ for all $k,j = 1,2$, is chosen as indicated in the label of \figurename~\ref{fig:sumrate} where for perfect CSI $\epsilon_{kj} = 0$. For each transmitter $k$, we generate the efficient beamforming vectors from \eqref{eq:ParOptBeam} with the parameters $\lambda_{k\ell}$ uniformly sampled in a $0.001$ step-length between zero and one. The set of rate tuples is calculated from the generated beamforming vectors and the maximum sum rate is found by grid search. For perfect CSI, full multiplexing gain of two is achieved with joint ZF transmission. Also, for the error scaling linearly with the inverse SNR as in \eqref{eq:errorScaling}, a multiplexing gain of two is achieved as is obtained in \eqref{eq:Sinfty3}. If the error scales slower than linearly with inverse SNR, as in $\epsilon_{kj} = 1/\sqrt[3]{\rho}$, the loss in multiplexing gain can be observed. For constant error of $\epsilon_{kj} = 0.3$ the multiplexing gain is one where only a single link operates as is derived in \eqref{eq:SumRateImpCSI2}. The transition from the operation of the two links to the operation of a single link can be noticed at around $30$ dB SNR in the maximum sum rate curve for constant error.

\subsection{Efficiency at Low SNR}
In order to study the spectral efficiency in the low SNR regime, two performance measures were introduced in \cite{Verdu2002}. The first measure is $\ebnof_{min}$ which is the minimum energy per bit required for reliable communication. The smaller $\ebnof_{min}$ is for a system the more reliable is its operation. The second performance measure is $S_0$ which is the slope of the spectral efficiency curve at $\ebnof_{min}$. Since $S_0$ describes the growth of the spectral efficiency curve from $\ebnof_{min}$, the larger $S_0$ is for a system the greater is the gain for increasing energy per bit from $\ebnof_{min}$. Here, we use these measures to determine optimal beamforming in our setting at low SNR.

The \emph{spectral efficiency} $\sC_\ell\pp{\ebnof_\ell}$ of link $\ell$ is defined as the ratio of transmission rate $R_\ell$ to bandwidth $B$ \cite[Section III]{Verdu2002}:
\begin{equation}
\sC_\ell\pp{\ebno_\ell} = \frac{R_\ell}{B} \label{eq:fl},
\end{equation}
\noindent where $\ebnof_\ell$ is the energy per bit for link $\ell$ normalized over the background noise. Specifically, $\ebnof_\ell = {P_\ell}/\pp{N_0 R_\ell},$ where $P_\ell$ is the transmitted power. The spectral efficiency function is directly related to the common capacity expression\footnote{Notice the difference in the notation for the spectral efficiency of link $\ell$, $\sC_\ell$, and the capacity of link $\ell$, $\cC_\ell$.} $\cC_\ell(\SNR)$ through $\sC_\ell\pp{\ebnof_\ell} = \cC_\ell(\SNR)$ for the SNR which solves \cite{Verdu2002}
\begin{equation}\label{eq:Cebno}
\ebno \cC_\ell(\SNR) = \SNR.
\end{equation}
\noindent At low SNR, the spectral efficiency function $\sC_\ell\pp{\ebnof_\ell}$ can be expressed as \cite{Verdu2002}
\begin{eqnarray}
	\sC_\ell \pp{\ebno_\ell} \approx \frac{S_{0,\ell}}{3 dB}\left( \ebno_\ell\Big|_{dB} - \ebno_{min,\ell}\Big|_{dB} \right), \label{eq:approxlow}
\end{eqnarray}
with
\begin{equation}\label{eq:minebno}
\ebno_{min,\ell} = \frac{\log_e 2}{\dot{\cC_\ell}(0)} \quad \text{and} \quad S_{0,\ell} = \frac{ 2 \left[ \dot{\cC_\ell}(0) \right]^2}{- \ddot{\cC_\ell}(0)},
\end{equation}%The closer $\ebnof_\ell$ gets to $\ebnof_{min,\ell}$ the better is the approximation in \eqref{eq:approxlow}.
\noindent where $\dot{\cC_\ell}$ and $\ddot{\cC_\ell}$ are the first and second derivatives, respectively, with respect to SNR. In \eqref{eq:minebno}, $\ebnof_{min,\ell}$ is the minimum energy per bit required for reliable communication for link $\ell$. It is shown in \cite[Proposition 4]{Larsson2008a} that the minimum energy per bit in \eqref{eq:minebno} is also valid in the MISO IFC. Furthermore, it is shown that single-user decoding at the receivers in the low SNR regime is optimal to minimize $\ebnof_{min,\ell}$. The rate of link $\ell$ from \eqref{eq:rate} is
\begin{equation}\label{Crho}
\cC_\ell(\SNR) = \log_2\pp{1+\frac{\SNR x^2_{\ell\ell}(\bw_\ell)}{1+\SNR \sum_{k \neq \ell} x^2_{k\ell}(\bw_k)}},
\end{equation}
\noindent where $x_{\ell\ell}(\bw_\ell)$ and $x_{k\ell}(\bw_k)$ are calculated in \eqref{eq:dir_gain} and \eqref{eq:int_gain}, respectively. For the purpose of calculating the minimum energy per bit and the wideband slope in \eqref{eq:minebno} for a link $\ell$, we have
\begin{equation}\label{eq:firstDeriv}
\dot{\cC}_\ell(0) = x^2_{\ell\ell}(\bw_\ell),
\end{equation}
\noindent and
\begin{equation}\label{eq:secDeriv}
\ddot{\cC}_\ell(0) = - x^2_{\ell\ell}(\bw_\ell) \pp{x^2_{\ell\ell}(\bw_\ell) + 2 \sum\nolimits_{k \neq \ell} x^2_{k\ell}(\bw_k)}.
\end{equation}
\noindent From \eqref{eq:firstDeriv} and \eqref{eq:secDeriv}, we can calculate the minimum energy per bit and the wideband slope in \eqref{eq:minebno} as
\begin{equation}\label{eq:minebno2}
\ebno_{\min,\ell} = \frac{\log_e 2}{x^2_{\ell\ell}(\bw_\ell)} \quad \text{and} \quad S_{0,\ell} = \frac{ 2 x^2_{\ell\ell}(\bw_\ell)}{x^2_{\ell\ell}(\bw_\ell) + 2 \sum\nolimits_{k \neq \ell} x^2_{k\ell}(\bw_k)}.
\end{equation}
\noindent The minimum energy per bit from \eqref{eq:minebno2} is minimized by maximizing the intended power gain $x^2_{\ell\ell}(\bw_\ell)$, i.e. with MRT beamforming. Thus, joint robust MRT beamforming is optimal in the low SNR regime to achieve the minimum energy per bit for reliable communication. For perfect CSI at the transmitter
\begin{equation}\label{eq:minEbNopCSI}
    \ebno_{\min,\ell}^\pCSI = \frac{\log_e 2}{\sabs{\bh_{\ell\ell}^H \bw^\mrt_\ell}} = \frac{\log_e 2}{\snorm{\bh_{kk}}}.
\end{equation}
For imperfect CSI at the transmitters, we have
\begin{equation}\label{eq:minEbNoiCSI}
    \ebno_{\min,\ell}^\iCSI = \frac{\log_e 2}{\pp{\abs{\ebh_{\ell\ell}^H \bw^\rmrt_\ell} - \norm{\mat{A}^H_{\ell\ell} \bw^\rmrt_\ell}\epsilon_{\ell\ell}}_+^2},
\end{equation}
\noindent where $\bw^\rmrt_\ell$ is obtained from \eqref{eq:MRT_rob}. Comparing \eqref{eq:minEbNoiCSI} to \eqref{eq:minEbNopCSI}, the loss due to imperfect CSI can be observed.

\begin{figure}[h]
\centering
[Figure 3 about here]
\end{figure}

In \figurename~\ref{fig:ebnoregion}, the $\ebnof$ regions for two links with perfect and imperfect CSI are plotted. The $\ebnof$ region with imperfect CSI is smaller and contained in the region with perfect CSI. For perfect and imperfect CSI, joint MRT achieves the joint minimum $\ebnof$.

The wideband slope from \eqref{eq:minebno2} of link $\ell$ for perfect and imperfect CSI are
\begin{equation}\label{eq:wideSlopePCSI}
    S_{0,\ell}^\pCSI = \frac{2 \sabs{\bh_{\ell\ell}^H \bw_\ell}}{\sabs{\bh_{\ell\ell}^H \bw_\ell} + 2 \sum_{k \neq \ell} \sabs{\bh_{k\ell}^H \bw_k}},
\end{equation}
\noindent and
\begin{equation}\label{eq:wideSlope}
    S_{0,\ell}^\iCSI = \frac{2 \pp{\abs{\ebh_{\ell\ell}^H \bw_\ell} - \norm{\mat{A}^H_{\ell\ell} \bw_\ell}\epsilon_{\ell\ell}}_+^2}{\pp{\abs{\ebh_{\ell\ell}^H \bw_\ell} - \norm{\mat{A}^H_{\ell\ell} \bw_\ell}\epsilon_{\ell\ell}}_+^2 + 2 \sum\limits_{k \neq \ell} \pp{\abs{\ebh_{k\ell}^H \bw_k} + \norm{\mat{A}^H_{k\ell} \bw_k}\epsilon_{k\ell}}^2},
\end{equation}
\noindent respectively. All jointly achievable wideband slopes for the links constitute a slope region \cite{Verdu2002a,Muharemovic2003}.

\begin{figure}[h]
\centering
[Figure 4 about here]
\end{figure}

In \figurename~\ref{fig:sloperegion}, the wideband slope regions for two links with perfect and imperfect CSI are plotted. The regions are generated by utilizing the parametrization of efficient beamforming vectors in \eqref{eq:ParOptBeam}. The wideband slope region with perfect CSI corresponds to the box in which the maximum wideband slope of two for both links can be achieved simultaneously (with ZF transmission). The wideband slope region with imperfect CSI is contained in that of perfect CSI, and the maximum wideband slope of two is achievable for one link only if the other link performs ZF, i.e. switches its transmission off. The wideband slopes with joint MRT are marked for perfect and imperfect CSI. These correspond to the slope of the spectral efficiency curve at $\ebno_{\min,\ell}$ as is given in \eqref{eq:approxlow}.

\begin{figure}[h]
\centering
[Figure 5 about here]
\end{figure}

In \figurename~\ref{fig:spect_ebno2}, the spectral efficiency of link $1$ in the two user MISO IFC setting used in \figurename~\ref{fig:ebnoregion} and \figurename~\ref{fig:sloperegion} is plotted for perfect and imperfect CSI. The curves are found by using the relation $\sC_\ell\pp{\ebnof_\ell} = \cC_\ell(\SNR)$ for the SNR which solves \eqref{eq:Cebno}. For the plot of $\sC_1 \pp{\ebnof}$ curves, joint MRT beamforming is assumed for all values of $\ebnof_1$. Joint MRT is however not Pareto efficient in mid or high SNR. From a game theoretic perspective \cite{Osborne1994}, joint MRT is the dominant strategy equilibrium (unique Nash equilibrium) of a strategic game between the links \cite{Larsson2008}. In other words, if the transmitters are noncooperative, they will jointly perform MRT beamforming. Accordingly, the performance plotted in \figurename~\ref{fig:spect_ebno2} is for two noncooperative links. The slopes of the tangents at $\ebno_{\min,1}^\pCSI$ and $\ebno_{\min,1}^\iCSI$ correspond to the wideband slopes in \figurename~\ref{fig:sloperegion} with joint MRT for perfect and imperfect CSI, respectively. The loss in minimum $\ebnof$ as well as wideband slope is apparent due to uncertainty in the channel information.

\section{Conclusions}\label{sec:conclusion}

We consider a $K$-user MISO IFC. The CSI is assumed to be perfect at the receivers but imperfect at the transmitters. Channel vector estimates at a transmitter include channel estimation errors which are assumed to be bounded in an elliptical region. The geometry of the uncertainty region associated with a channel vector estimate is known at the transmitter. In this setting, robust beamforming optimizes worst-case power gains at the receivers. We derive the worst-case intended and interference power gains at the receivers and formulate accordingly the worst-case achievable rates for the links. Afterwards, we characterize the robust beamforming vectors necessary to operate at any Pareto optimal point in the robust achievable rate region. The efficient beamforming vectors of each transmitter are found as a solution of a SOCP which can be solved efficiently. The spectral efficiency of the multi-link system with imperfect channel state information is analyzed in the high and low SNR regime. At high SNR, achieving full multiplexing gain with zero forcing transmission using the channel vector estimates requires the channel estimation error to reduce linearly with the SNR. If the error does not depend on SNR, single-user transmission is optimal. In the low SNR regime, it is shown that joint robust maximum ratio transmission optimizes the minimum energy per bit for reliable communication.

%\section*{References}

%\bibliographystyle{elsarticle-num}
%\bibliography{references}
%\input{sections/bio}
\pagebreak[4]

\begin{figure}[h]
  \centering
  \includegraphics[width=\linewidth,clip]{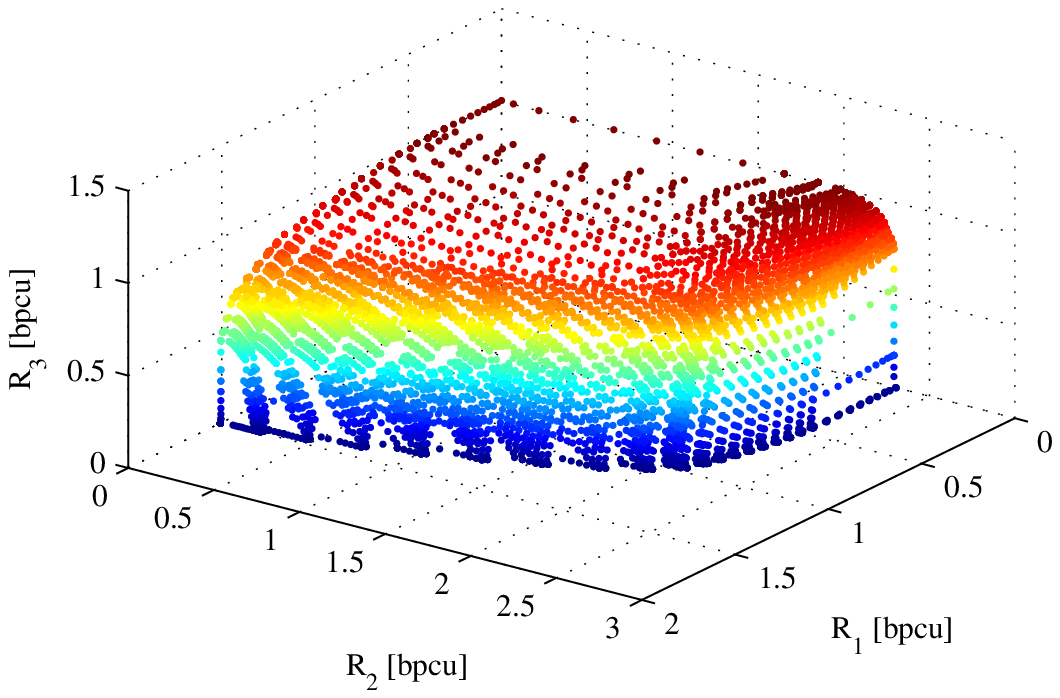}
  \caption{Pareto boundary of a three-user robust rate region at $0$ dB SNR.}\label{fig:rateregion}
\end{figure}

\pagebreak[4]

\begin{algorithm}[h]
    \IncMargin{2em}
    %\KwResult{Multiplexing Gain}
    \KwIn{$[N_1,\ldots,N_K]$}
    Sort $[N_1,\ldots,N_K]$ in weakly decreasing order to $[\tilde{N}_1,\ldots,\tilde{N}_K]$ such that $\tilde{N}_i \geq \tilde{N}_{i+1}$\;
    \For {$k = 1,\ldots,K$}{
    \If {$\tilde{N}_k < k$}{
     $m^* = k - 1$\;
     \textbf{break}\;
    }
    }
    \KwOut{$m^*$}
    \caption{Algorithm to calculate the maximum multiplexing gain.}
    \label{alg:multgain}
\end{algorithm}%

\pagebreak[4]

\begin{figure}[h]
  \centering
  \includegraphics[width=\linewidth,clip]{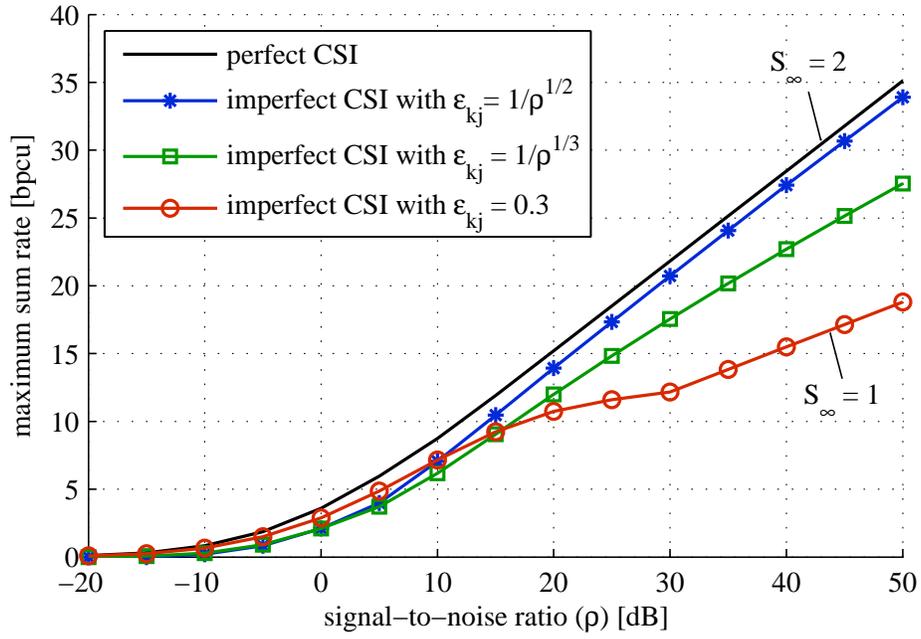}
  \caption{Comparison of the maximum sum rate of two links for different models of channel estimation errors $\epsilon_{kj}, k,j=1,2$.} \label{fig:sumrate}
\end{figure}

\pagebreak[4]

\begin{figure}[h]
  \centering
  \includegraphics[width=\linewidth,clip]{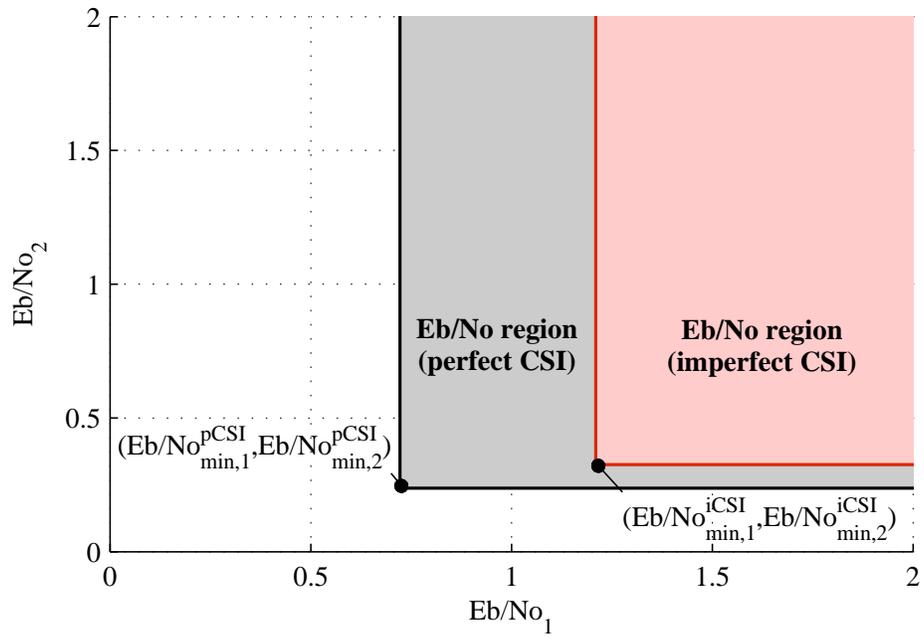}
  \caption{Plot of minimum energy per bit regions of two links for perfect and imperfect CSI.}\label{fig:ebnoregion}
\end{figure}

\pagebreak[4]

\begin{figure}[h]
  \centering
  \includegraphics[width=\linewidth,clip]{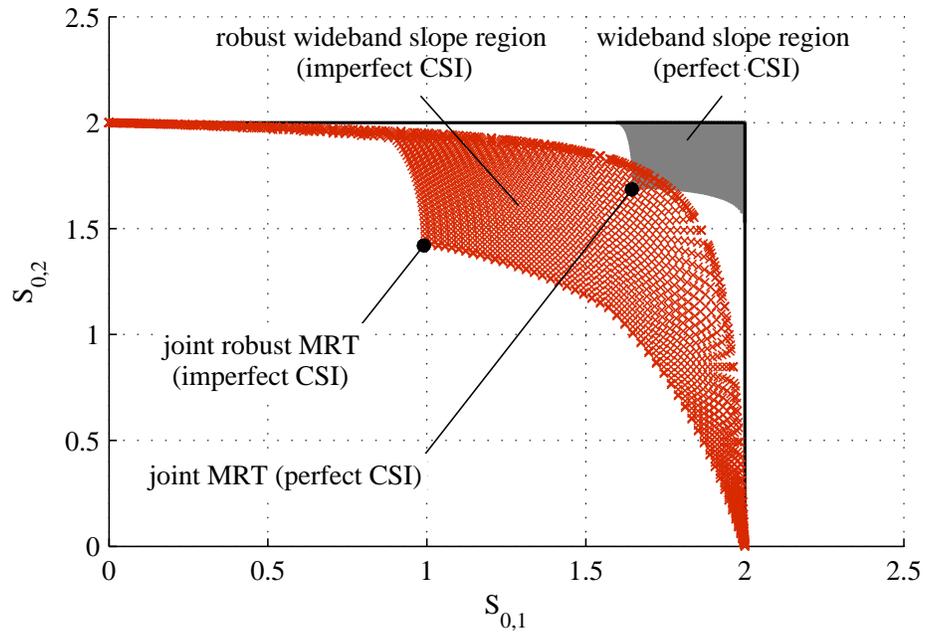}
  \caption{Plot of wideband slope regions of two links for perfect and imperfect CSI.}\label{fig:sloperegion}
\end{figure}

\pagebreak[4]

\begin{figure}[h]
  \centering
  \includegraphics[width=\linewidth,clip]{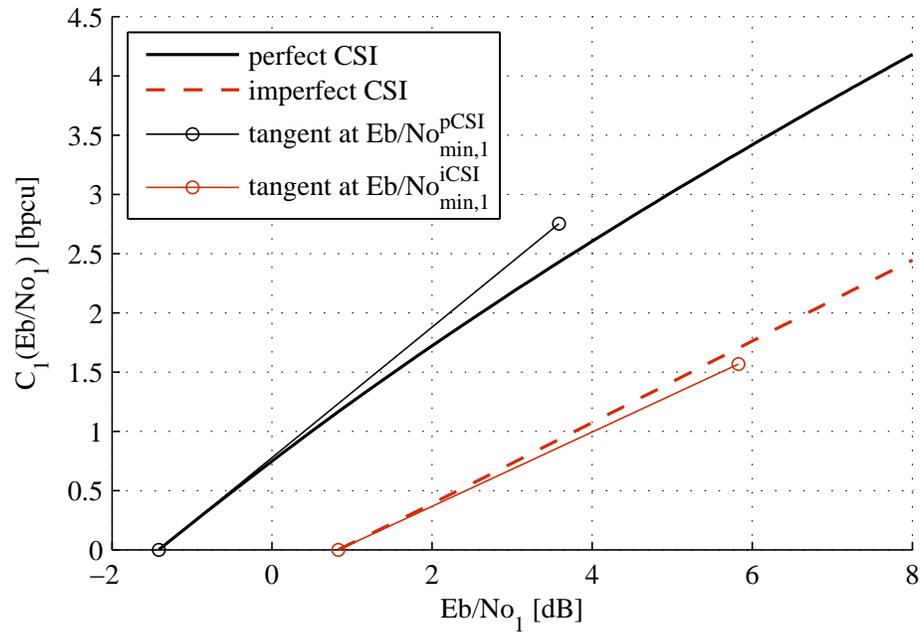}
  \caption{Plot of the spectral efficiency of link $1$ for perfect and imperfect CSI in a two-user MISO IFC.}\label{fig:spect_ebno2}
\end{figure}
\end{document}